\documentclass{article}
\usepackage{amsmath,amssymb,amsbsy,amsthm,setspace}

\let\epsilon\varepsilon
\let\phi\varphi

\let\epsilon\varepsilon

\newtheorem*{lemma*}{Lemma}
\newtheorem{theorem}{Theorem}

\newtheorem{corollary}{Corollary}

\begin{document}



\title{The Vernam cipher is
 robust to small deviations from randomness}

 \author{Boris Ryabko\\ 
Siberian State University of Telecommunications and Information Sciences,\\
Institute of Computational Technology of Siberian Branch of Russian \\ Academy of Science,  Novosibirsk, Russia\\
   boris@ryabko.net}


\date{}

\maketitle


\begin{abstract}

The Vernam cipher (or one-time pad)  has played an important rule in cryptography because it is a perfect secrecy system. 
For example, if an English text (presented in binary system) $X_1 X_2 ... $ is enciphered 
  according to the 
formula $Z_i = (X_i + Y_i) \mod 2 $, where $Y_1 Y_2 ...$ is a key sequence generated by the Bernoulli source
with equal probabilities of 0 and 1, anyone who knows   $Z_1 Z_2  ... $ has no information about $X_1 X_2 ... $ 
without the knowledge of the key
$Y_1 Y_2 ...$.
(The best strategy is to guess $X_1 X_2 ... $ not paying attention to $Z_1 Z_2  ... $.)

But what should one say about secrecy of  an analogous method where the key sequence $Y_1 Y_2 ...$ is generated by the Bernoulli 
source with a small bias, say, $P(0) = 0.49, $ $ P(1) = 0.51$? 
To the best of our knowledge, there are no theoretical estimates for the secrecy of such a  system,
as well as for the general case where $X_1 X_2 ... $ (the plaintext) and  key sequence are described by stationary ergodic
processes. 
We consider the running-key ciphers where the plaintext  and the key are generated by stationary ergodic sources
and show how to estimate the secrecy of such systems. In particular, it is shown that, in a certain sense, the Vernam cipher
is robust to small deviations from randomness.  
\end{abstract}


\section{Introduction}
We consider the classical problem of transmitting secret messages from Alice (a sender) to Bob (a receiver) 
via an open channel which can be accessed by Eve (an adversary). It is supposed that Alice and Bob (and nobody
else) know a so-called key $K$ which is a word in a certain alphabet.
Before transmitting a message   Alice  encrypts it. In his turn, Bob,
after having received the encrypted message (ciphertext), decrypts it to recover the initial text (plaintext).

We consider so-called running-key ciphers where the plaintext $X_1 ... X_t$, the key sequence $Y_1 ... Y_t$ 
and ciphertext $Z_1 ...  Z_t$ belong to one alphabet $A$ (without loss of generality we suppose that 
$A = \{0,1, ... , n - 1\}$, where $n \ge 2$. The $i-th$ letter of the ciphertext is defined by 
$
Z_i = c(X_i,Y_i), \, $ $ \, i = 1, ... , t
$, whereas the deciphering rule is by  $
X_i = d(Z_i,Y_i), \, $ $ \, i = 1, ... , t
$, i.e. $d(e(X_i,Y_i), Y_i) = X_i$. Here $c$ and $d$ are functions called coder and decoder, correspondingly. 
Quite often the following particular formula are used
\begin{equation}\label{v1}
Z_i = (X_i + Y_i) \mod n \, , \, \, \, \,  \, X_i = (Z_i - Y_i) \mod n \, , 
\end{equation}
 i.e. $ c(X_i,Y_i) = (X_i + Y_i) \mod n ,  $ $\, \,  d(Z_i,Y_i) = (Z_i - Y_i) \mod n  $.
 In a case of two-letter alphabet (\ref{v1}) can be  presented as follows:
\begin{equation}\label{v2}
Z_i = (X_i \oplus Y_i)   \, , \, \, X_i = (Z_i \oplus Y_i)  
\end{equation}
where $a \oplus b \, = (a+b) \mod 2$.

 It is important to note that we consider a so-called unconditional (or information-theoretical)  security.
That is, the cipher is secure even when Eve has unlimited computing power. Roughly speaking, if the unconditionally
secure cipher is used,  Eve has many 
highly probable possible versions of a plaintext and, hence,  cannot choose the real plaintext from them.  
   The following informal consideration helps to understand the main idea of an approach considered later: 
   Let there be two  unconditionally
secure ciphers
which can be applied to one plaintext. 
Imagine, that for the first cipher Eve has 10 equiprobable possible deciphering texts
whose  overall probability equals 0.999, 
whereas for the second cipher  there are 100
equiprobable deciphering texts with the same overall probability. Obviously, the second system is more preferable,
because the uncertainty of Eve is much larger for the second system.
   This informal consideration  is quite popular in cryptography \cite{RW,RF,S} and 
 we will estimate the security of a cipher by the logarithm of the total number of 
   (almost) equiprobable possible deciphering texts whose overall probability is close to 1.

 The running-key cipher (\ref{v1}) is called the Vernam cipher  (or one-time pad) if  any word 
 $k_1 ... k_t$, $k_i \in A$, is used as the key word with probability $n^{-t}$,
 i.e. $ P(Y_1 ... Y_t = k_1 ... k_t) = n^{-t}$ for any 
 $ k_1 ... k_t \in A^t$.  In other words, we can say that the key letters are 
 independent and identically distributed (i.i.d.) and probabilities of all letters  are equal. 
 
The Vernam cipher is one of the most popular among the unconditionally
secure   running-key ciphers. It has played an important rule in cryptography, especially since 
C.Shannon proved that this cipher is perfectly secure \cite{S}. 
That is, the ciphertext  $Z_1 ...  Z_t$ gives absolutely no additional information about the plaintext
 $X_1 ... X_t$. This fact can be interpreted as follows:
   a priori probability of a plaintext  is the same as a posteriori probability of 
a plaintext   given the corresponding ciphertext \cite{S}. Using Shannon entropy, 
it can be  expressed by the following equation  $h(X_1 ... X_t) = h(X_1 ... X_t|Z_1 ...  Z_t)$, 
 where $h(X_1 ... X_t)$ and $h(X_1 ... X_t|Z_1 ...  Z_t)$ are the entropy of the plaintext and 
  the conditional entropy of the plaintext given the ciphertext $Z_1 ...  Z_t$, correspondingly
  (they will be defined below).   
   For example, if one uses the Vernam cipher (\ref{v2}) to cipher an  English text presented, say,  
   in standard 7-bit binary ASCII, Eve can  try to guess the plaintext not paying attention 
   on the ciphertext.
   
   It was shown by Shannon that  any perfectly secure system 
    must use the secret key
whose length equals the plaintext length. 
That is why many authors considered  
 the problem of security of systems where either the length of the key
  or its entropy  is less than the length (or entropy) of the plaintext,
  see,  for example, \cite{ahl,he,m1,RW,RF,S} and reviews therein. 
  But, in spite of numerical 
  papers, some seemingly natural  questions  are  still open. 
  For example, what can we say about secrecy of the system  (\ref{v2}) 
 where it is applied to an English text (in binary 
presentation) and the key sequence is generated by the Bernoulli source with a small 
bias, say, $P(Y_i = 0) = 0.51, $ $P(Y_i = 0) = 0.49$. (Informally, it is ``almost'' Vernam cipher).
To the best of our knowledge, there are no theoretical estimates for the security of such a  system,
as well as for the general case where the plaintext and  key are described as stationary ergodic
processes. 

  In this paper we consider  this problem for 
 running-key ciphers (\ref{v1})  in a case where the plaintext $X_1 ... X_t$ and 
 the key sequence $Y_1 ... Y_t$ are independently generated by stationary ergodic sources
 and the entropy of the key can be less than maximally possible value $\log n$ per letter (here and below $\log \equiv \log_2 $).  
 The goal of the paper is to find simple estimates 
   of secrecy for such systems.  We would like to emphasize that the unconditional secrecy is meant, i.e. it is supposed that
   Eve has unlimited computational power and unlimited time for computations.

   It is worth noting that Shannon in his famous paper \cite{S} mentioned that 
  the problem of deciphering of a ciphertext and the problem of signal denoising
  are very close from mathematical point of view. In this paper we use some results obtained  in 
  \cite{rr11} considering the problem of denoising.

\section{Preliminaries} 

We consider the case where the plaintext $X = X_1,X_2,\dots$ and 
and the key sequence  $Y_1,Y_2,\dots$  are independently generated by 
stationary  ergodic processes with the finite alphabets $A = \{0,1, ... , n - 1\}$,  $n \ge 2$.

The $m-$order Shannon entropy and the limit Shannon entropy are defined as follows:
\begin{equation}\label{ent}
 h_m(X) = -  \frac{1}{m+1}\sum_{u \in A^{m+1} } P_X(u) \log P_X(u) , \ \  
h(X) = \lim_{m \rightarrow \infty} h_m(X)
\end{equation}
where 
$m \ge 0$ , $P_X(u)$ is the probability that $X_1 X_2 ... X_{|u|} $ $ = u$ (this limit always exists, see, 
for ex., \cite{Cover:06, Gallager:68}). 
Introduce also the  conditional Shannon entropy 
\begin{equation}\label{entc}
h_m(X|Z) = h_m(X,Z) - h_m(Z), \, \, h(X|Z) = \lim_{m \rightarrow \infty} h_m(X|Z)
\end{equation}

The  Shannon-McMillan-Breiman theorem  
for  conditional
entropies  can be stated as follows. 
\begin{theorem}[Shannon-McMillan-Breiman]\label{th:smb}
  $\forall \varepsilon >0, \forall\delta > 0$, for almost all \\ $Z_1,Z_2,\dots$ 
  there exists $n'$ such that 
  if $n > n'$  then 
\begin{equation}\label{smb}
   P\left\{ \left| -  \frac{1}{n} \log P(X_1..X_n|Z_1..Z_n) 
 - h(X|Z) \right| < \epsilon \right\} \ge 1-\delta. 
\end{equation}
\end{theorem}
The proof can be found in \cite{Cover:06,Gallager:68}.

\section{Estimations of secrecy}

\begin{theorem}\label{th:1}
Let a plaintext    $X = X_1 X_2,\dots$ and the key sequence  $Y= Y_1 Y_2,\dots$ be 
stationary ergodic processes with a finite alphabet $A = \{0,1, ... , n - 1\}$,  $n \ge 2$,
and let a running-key cipher be applied to $X$ and $Y$ and $Z = Z_1,Z_2,\dots$
be the ciphertext. Then, for any $\varepsilon > 0$ and  $\delta > 0$ 
there is such an integer $n'$ that, with probability 1,  
for any $t > n'$ and $Z = Z_1,Z_2,\dots Z_t$
there exists the set $ \Psi (Z)$ for which  the following properties are valid:

i) $ P( \Psi (Z) ) > 1 - \delta$

ii) for any $X^1 = X_1^1, \dots ,  X^1_t$, $X^2 = X_1^2, \dots ,  X^2_t$  from $ \Psi (Z)$
$$
P\left\{  \frac{1}{t} \left|   \log P(X^1|Z) 
 -  \log P(X^2|Z)  \right| < \epsilon \right\}
$$

iii) $ \liminf_{t \rightarrow \infty } \frac{ 1}{t} \log  |\Psi (Z) |  \ge h(X|Z) \, . $  

\end{theorem}
\begin{proof}
According to Shannon-McMillan-Breiman theorem for   any $\epsilon > 0, \delta > 0$ and almost all  $Z_1,Z_2,\dots$
there exists such $n'$ that  for $t > n'$ 
\begin{equation}\label{smb1}
   P\left\{ \left| -  \frac{1}{t} \log P(X_1 X_2 ... X_t|Z_1 Z_2 ... Z_t ) 
 - h(X|Z) \right| < \epsilon /2  \right\} \ge 1-\delta. 
\end{equation}
Let us define 
\begin{equation}\label{psi}
    \Psi (Z) = \{  X = X_1 X_2 ... X_t : P(X_1X_2 ... X_t|Z_1 Z_2 ... Z_t ) 
 - h(X|Z) | < \epsilon /2    \} \, . 
\end{equation}
The first property i) immediately follows from (\ref{smb1}). 
In order to prove ii),  note that 
for any $X^1 = X_1^1, \dots ,  X^1_t$, $X^2 = X_1^2, \dots ,  X^2_t$  from $ \Psi (Z)$ 
we obtain from (\ref{smb1}), (\ref{psi})
$$
 \frac{1}{t} \left|   \log P(X^1|Z) 
 -  \log P(X^2|Z)  \right | \le  \frac{1}{t} \left  | \log P(X^1|Z) -h(X|Z) \right | 
 $$ $$ + 
  \frac{1}{t} \left| \log P(X^2|Z) -  h(X|Z)  \right| < \epsilon /2  + \epsilon /2  = \epsilon \, .
$$
From (\ref{psi}) and  the property i) we obtain the following:
$ \,  |\Psi (Z) | > ( 1- \delta) 2^{ t \, (h(X|Z) - \epsilon)}   \, .$
Taking into account that it is valid for any  $\epsilon > 0, \delta > 0$ 
and $t > n'$, we obtain 
iii). 
\end{proof}
So, we can see that  the set of possible decipherings  $ \Psi (Z)$ grows exponentially,
its total probability is close to 1 and probabilities 
of words from this set are close to each other.

 Theorem 2 gives a possibility to estimate an uncertainty of a cipher based on the conditional entropy 
$h(X|Z)$. Sometimes it can be difficult to calculate this value because it  requires knowledge of the conditional probabilities.
In this case the following simpler estimate can be useful. 
\begin{corollary}\label{th:c} For almost all $Z_1 Z_2 ... $
  $$ \liminf_{t \rightarrow \infty } \frac{ 1}{t} \log  |\Psi (Z) |  \ge h(X) + h(Y) - \log n \, . $$ 
 \end{corollary}
\begin{proof} 
From the well-known in Information Theory equation 
$h(X,Z) = h(X) + h(Z|X) $ (see  \cite{Cover:06,Gallager:68}) we obtain the following: 
$$
h(X|Z) = h(X,Z) - h(Z) = h(Z|X) + h(X) - h(Z) .
$$
Having taken into account that $\max h(Z) = \log n$ (\cite{Cover:06,Gallager:68}),
where $n$ is the number of alphabet letters, we can derive from the latest equation 
that 
$h(X|Z) \ge h(Z|X) + h(X) - \log n.$ 
The definition of the running-key cipher (\ref{v1}) shows that $h(Z|X) = h(Y)$.  Taking into account
two latest inequalities 
 and the third statement iii)  of Theorem 2 we obtain 
the statement of the corollary. 
\end{proof}

\textbf{Comment.} 
In Information Theory the difference between maximal value of the entropy and real one quite often is called the redundancy. 
Hence, from the corollary we have new following presentations for the value $\frac{ 1}{t} \log  |\Psi (Z) |$:
\begin{equation}\label{re}
\liminf_{t \rightarrow \infty } \frac{ 1}{t} \log  |\Psi (Z) |  \ge h(X) - r_Y \, , \, \, 
\, \, \, 
\liminf_{t \rightarrow \infty } \frac{ 1}{t} \log  |\Psi (Z) |  \ge h(Y) - r_X \, , 
$$
$$
\liminf_{t \rightarrow \infty } \frac{ 1}{t} \log  |\Psi (Z) |  \ge \log n - (r_X + r_Y) \, ,
\end{equation}
where $ r_Y  = \log n - h(Y)$ and $ r_X  = \log n - h(X)$ are the corresponding redundancies.

Those inequalities confirm  the well-known 
in cryptography and Information Theory fact that reduction of the redundancy improves the safety of ciphers.

Let us return to the first question of this note about the Vernam cipher with a biased key sequence.
More precisely, let there be a plaintext $X_1 X_2 ... $, $X_i \in \{0,1\} $  and the key sequence $Y_1 Y_2 ... $, $Y_i \in \{0,1\} $,
generated by 
a source whose entropy $h(Y)$ is less then 1. ($h(Y) = 1$ if and only if $Y_1 Y_2 ... $ generated by
the Bernoulli
source with letter probabilities $P(0) = P(1) = 0.5$, \cite{Cover:06,Gallager:68}).
 From  (\ref{re}) we can see that the size
of the set $\Psi (Z)$ of high-probable 
possible decipherings  grows exponentially with exponent grater than 
$h(X) - r_Y$, where $r_Y = 1 - h(Y)$.
So, if $r_Y$ goes to 0,  the size of the set of possible probable decipherings trends to the size of this set for the case 
of ``pure'' Vernam cipher. Indeed,   if $h(Y) = 1$ and, hence, $r_Y = 0$,  the set $\Psi (Z)$ of high-probable 
possible decipherings  grows exponentially with exponent $h(X)$, as it should be for the Vernam cipher.
For example, it is true for the case where the key sequence 
$Y_1 Y_2 ... $ is generated by the Bernulli source with biased probabilities, say 
$P(0) = 0.5 - \tau, \,
P(1) = 0.5 + \tau$, where $\tau$ is a small number. If $\tau$ goes to 0, 
the redundancy $r_Y$ goes to 0, too, and we obtain the Vernam cipher. 
So, we can informally say that the Vernam cipher is
 robust to small deviations from randomness.




\end{document}